\documentclass{article}
\usepackage[utf8]{inputenc}

\usepackage{amsmath}
\usepackage{amsthm}
\usepackage{cleveref}
\usepackage{enumitem}
\usepackage{color}
\usepackage{tikz}
\usepackage{amsfonts}%for \mathbb{N}
\usepackage{todonotes}
\newtheorem{algorithm}{Algorithm}
\newtheorem{definition}{Definition}
\newtheorem{lemma}{Lemma}
\newtheorem{theorem}{Theorem}
\newtheorem{corollary}{Corollary}

\title{Extending Lattice linearity for Self-Stabilizing Algorithms}
\author{Arya Tanmay Gupta\and Sandeep S Kulkarni}
\date{\texttt{\{atgupta, sandeep\}@msu.edu}\\Computer Science and Engineering, Michigan State University}

\begin{document}

\maketitle

\begin{abstract}
    In this article, we focus on extending the notion of lattice linearity to self-stabilizing programs. Lattice linearity allows a node to execute its actions with old information about the state of other nodes and still preserve correctness. It increases the concurrency of the program execution by eliminating the need for synchronization among its nodes. 
    
    The extension --denoted as eventually lattice linear algorithms-- is performed with an example of the service-demand based minimal dominating set (SDDS) problem, which is a generalization of the dominating set problem; it converges in $2n$ moves. Subsequently, we also show that the same approach could be used in various other problems including minimal vertex cover, maximal independent set and graph coloring. 
\end{abstract}


\section{Introduction}\label{section:introduction}

In a distributed program, a node cooperates with other nodes to solve the problem at hand such as leader election, mutual exclusion, tree construction, dominating set, independent set, etc. There are several models for such distributed programs. These can be broadly classified as message passing programs or shared-memory programs. In message passing programs, nodes do not share memory. Rather, they communicate with each other via messages. On the other hand, the shared-memory model allows a node to read the memory of other nodes to solve the given problem.

Implementation of such 
shared memory programs introduces several challenges to allow a node to read the state of its neighbors in a consistent fashion. One solution in this context is that the nodes only execute in a coordinated manner where when a node is activated by the scheduler, it reads the variables of other nodes and updates its own state. Furthermore, the scheduler needs to ensure that any \textit{conflicting} nodes are not activated at the same time. 
This approach, however, is expensive and requires synchronization among nodes. 

To alleviate the issue of consistency while reading remote variables, 
Garg (2020) \cite{Garg2020} introduced lattice linear predicate detection in combinatorial optimization problems.
In \cite{Garg2020}, it is shown that when an algorithm exploits lattice linearity of the underlying problem, it preserves correctness even if nodes execute with old information. However, this work assumes that the algorithm begins in a specific initial state and, hence, is not applicable for self-stabilizing algorithms since a self-stabilizing algorithm guarantees that starting from an arbitrary state, the algorithm reaches a legitimate state (\textit{invariant}) and remains there forever. With this intuition, in this work, we extend the results in \cite{Garg2020} to self-stabilizing algorithms.

We proceed as follows. We begin with the problem of service-demand based minimal dominating set (SDDS) which is a generalization of the dominating set problem. We devise a self-stabilizing algorithm for SDDS. We scrutinize this algorithm and disassemble it into two parts, one of which satisfies the lattice linearity property of \cite{Garg2020} if it begins in a \textit{feasible} state.  Furthermore, we show that the second part of the algorithm ensures that the algorithm reaches a \textit{feasible} state. We show that the resulting algorithm is self-stabilizing, and the algorithm 
has 
\textit{limited-interference} property (to be discussed in \Cref{subsection:ds-eventual}) due to which it is tolerant to the nodes reading old values of other nodes. 

We also demonstrate that this approach is generic. It applies to various other problems including vertex cover, independent set and graph coloring. 

\subsection{Contributions of the article}

\begin{itemize}
    \item We present a self-stabilizing algorithm for the minimal SDDS problem. The algorithm can be modified to solve other generalizations of dominating set present in the literature. The algorithm converges in $2n$ moves, which is an improvement over the other algorithms in the literature.

%\item We introduce the class of lattice linear and eventually lattice algorithms.
\item We extend the notion of \textit{lattice linear predicate detection} from \cite{Garg2020} to introduce the class of \textit{lattice linear self-stabilizing algorithms} and \textit{eventually lattice linear self-stabilizing algorithms}.
Such algorithms allow the program to converge even when the nodes read old values. This is unlike the algorithms presented in \cite{Garg2020} where it is required that (1) the problems have only one optimal state, and (2)  the program needs to start in specific initial states.

\item Our solution to SDDS can be extended to other problems including minimal vertex cover, maximal independent set and graph coloring problems.
The resulting algorithms are eventually lattice linear and can be modified to lattice linear self-stabilizing algorithms. 
\end{itemize}

\subsection{Organization of the article}

In \Cref{section:literature}, we discuss the related work in the literature. In \Cref{section:preliminaries}, we discuss some notations and definitions that we use in the article.
In \Cref{section:sdds-algorithm}, we describe the algorithm for the service-demand based dominating set problem. In \Cref{section:sdds}, we analyze the characteristics of the algorithm and show that it is eventually lattice linear. In \Cref{section:other-examples}, we use the structure of eventually lattice linear self-stabilizing algorithms to develop algorithms for vertex cover, independent set and graph coloring problems.
Finally, we conclude the article in \Cref{section:conclusion}.

\section{Literature study and discussion on our contribution}\label{section:literature}
Self-stabilizing algorithms for the minimal dominating set problem have been proposed in several works in the literature, for example, in \cite{Xu2003,Hedetniemi2003,Turau2007,GODDARD2008,Chiu2014}. The best convergence time among these works is $4n$ moves.

Other variations of the dominating set problem are also studied. Fink and Jacobson (1985) \cite{Fink1985} proposed the minimal $k$-dominating set problem; here, the task is to compute a minimal set of nodes $\mathcal{D}$ such that for each node $v\in V(G)$, $v\in\mathcal{D}$ or there are at least $k$ neighbors of $v$ in $\mathcal{D}$. When $k=1$, the definition of $\mathcal{D}$ here is same as that in the general dominating set problem. Kamei and Kakugawa (2003) \cite{Kamei2003} proposed self-stabilizing algorithm for tree networks under central and distributed schedulers for the minimal $k$-dominating set; the converge time is $n^2$ moves. Kamei and Kakugawa (2005) \cite{Kamei2005} have proposed a self-stabilizing distributed algorithm which converges in $2n+3$ rounds; their algorithm runs on synchronous daemon.

A generalization of the dominating set problem is described in Kobayashi et al. (2017) \cite{Kobayashi2017}. This article assumes the input to include wish sets (of nodes) for every node. For each node $i$, either $i$ should be in the dominating set $\mathcal{D}$ or at least one of its wish set must be a subset of $\mathcal{D}$. In this case, the input size may be exponential. 
The nodes require to read the latest values of other nodes.

Self-stabilizing algorithms for the vertex cover problem
has been studied in Kiniwa (2005) \cite{Kiniwa2005}, Astrand and Suomela (2010) \cite{Astrand2010}, and Turau (2010) \cite{TURAU2010}. A survey of self-stabilization algorithms on independence, domination and coloring problems can be found in Guellati and Kheddouci (2010) \cite{Guellati2010}.

Garg (2020) \cite{Garg2020} studied the exploitation of lattice linear predicates in several problems to develop parallel processing algorithms. Lattice linearity ensures convergence of the system to an optimal solution while the nodes perform executions parallely and are allowed to do so without coordination, and are allowed to perform executions based on the old values of other nodes. Garg (2021) \cite{Garg2021} introduces lattice linearity to dynamic programming problems such as longest increasing subsequences and knapsack problem. In this approach, the lattice arises from the computation. We are going to pursue a similar goal for self-stabilizing problems.

Our SDDS algorithm uses local checking to determine if it is in an inconsistent state and local correction to restore it. Thus, it differs from \cite{Varghese1992} where global correction in the form of reset is used. Local detection and correction is also proposed in \cite{Arora1993,Leal2004}. The key difference with our work is that we are focusing on scenarios where local detection can be performed without requiring coordination with other nodes. 

\section{Preliminaries}\label{section:preliminaries}

\subsection{Modeling Algorithms}
Throughout the article, we denote $G$ to be an arbitrary graph on which we apply our algorithms. $V(G)$ is the vertex-set and $E(G)$ is the edge-set of $G$. In $G$, for any node $i$, $Adj_i$ is the nodes connected to $i$ in $G$, and $N_i=Adj_i\cup\{i\}$.
$deg(i)$ denotes the degree of node $i$.

Each node $i$ is associated with a set of variables. The algorithm is written in terms of rules, where each \textit{rule} for process $i$ is of the form $guard\longrightarrow action$ where \textit{guard} is a proposition over variables of 
some nodes which may include the variables of $i$ itself along with the variables of other nodes.
If any of the guards hold true for some node, we say that the node is \textit{enabled}. 
As the algorithm proceeds, we define a \textit{move} with reference to a node to be an action in which it changes its state. A \textit{round} with reference to a scheduler is a minimum time-frame where each node 
is given a chance to evaluate its guards and take action (if some guard evaluates to true) at least once.

An algorithm is \textit{silent} if no node is enabled when $G$ reaches an optimal state (we describe the respective optimal states as we discuss the problems in this article).

\textbf{Scheduler/Daemon.} A \textit{central scheduler/daemon} is a scheduler which chooses only one node to evaluate its guards in a time-step and execute the corresponding action. A \textit{distributed scheduler/daemon} chooses an arbitrary subset of nodes of $V(G)$ in a time-step to evaluate their guards and execute the corresponding actions respectively. A \textit{synchronous scheduler/daemon} chooses all the nodes in $V(G)$ in each time-step to evaluate their guards and execute the corresponding actions respectively together.

\textbf{Read/Write Model.}
In the read/write model, we partition the variables of a node as \textit{public} variables that can be read by others, and \textit{private} variables that are only local to that node. 
In this model, the rules of the node are allowed to be either:
\begin{enumerate}[label=(\alph*)]
    \item \textit{read rules}, where any node $i$ is allowed to read the public variables of one or more or all of its neighbors and copies them into a private variable of $i$, or 
    \item \textit{write rules}, where $i$ reads only its own variables to update its public variables.
\end{enumerate}

\subsection{Lattice Linear Predicates}

\textit{Lattice Linearity} \cite{Garg2020} of a problem is a phenomenon by which all the state vectors of a global state of a system $G$ form a distributive lattice. The predicate which defines an optimal state of the problem (under which such a lattice forms) is called a \textit{lattice linear predicate}. In such a lattice, if the state of a system $G$ is false according to the predicate, then at least one node $i\in V(G)$ can be identified such that it is \textit{forbidden}, that is, in order for $G$ to reach an optimal state, $i$ must change its state. Since this article studies self-stabilization problems, we define the predicate to be an optimal state with respect to the respective problems.

\subsection{The communication model}

The nodes of a graph communicate via shared memory. In each action, a node reads the values of its distance-k neighbors (where value of $k$ depends upon the specific algorithm) and updates its own state. We make no assumptions about atomicity with respect to reading the variables. In other words, while one node is in the middle of updating its state, its neighbors may be updating their owns state as well. In turn, this means that when node $i$ changes its state (based on state of node $j$) it is possible that the state of node $j$ has changed. In other words, $i$ is taking an action based on an old value of node $j$.  
Therefore, our algorithms will run equally well in a message passing model (with distance-$k$ flooding), without the requirement of synchronization or locks.

\section{Service-Demand based Dominating Set}\label{section:sdds-algorithm}

In this section, 
we introduce a generalization of the dominating set problem,
the service-demand based dominating set problem and describe an algorithm to solve it. 

\begin{definition}\textbf{Service-demand based dominating set problem (SDDS).}
    In the minimal \textit{service-demand based dominating set} problem, the input is a graph $G$ and a set of services $S_i$ and a set of demands $D_i$ for each node $i$ in $G$; the task is to compute a minimal set $\mathcal{D}$ such that for each node $i$,
    \begin{enumerate}
        \item either $i\in \mathcal{D}$, or
        \item for each demand $d$ in $D_i$, there exists at least one node $j$ in $Adj_i$ such that $d\in S_j$ and $j\in \mathcal{D}$.
    \end{enumerate}
\end{definition}
In the following subsection, we present a self-stabilizing algorithm for the minimal SDDS problem.
Each node $i$ is associated with variable $st.i$ with domain $\{IN, OUT\}$. $st.i$ defines the state of $i$. We define $\mathcal{D}$ to be the set $\{i\in V(G): st.i=IN\}$. 

\subsection{Algorithm for SDDS problem}\label{subsection:ds-general-algorithm}

The list of constants stored in each node is described in the following table. For a node $i$, $D_i$ is the set of demands of $i$, $S_i$ is the set of services that $i$ can provide to its neighbors. $D_i$ and $S_i$ are provided as part of the input.

\begin{center}
    \begin{tabular}{|l|l|}
        \hline
        Constant & What it stands for\\
        \hline
        $D_i$ & the set of demands of node $i$.\\
        $S_i$ & the set of services provided by node $i$.\\
        \hline
    \end{tabular}
\end{center}

The list of macros stored in each node is described in the following table.
Recall that $\mathcal{D}$ is the set of nodes which currently have the state as $IN$. \textsc{Satisfied}($i$) is true if $i\in\mathcal{D}$ or each demand $d$ in $D_i$ is being served by some node $j$ in $Adj_i$.
If \textsc{Removable}($i$) is true, then $\mathcal{D}\setminus\{i\}$ is also a dominating set given that $\mathcal{D}$ is a dominating set.
\textsc{Dominators-Of}($i$) is the set of nodes that are (possibly) dominating node $i$: if some node $j$ is in \textsc{Dominators-Of}($i$), then there is at least one demand $d\in D_i$ such that $d\in S_i$.
We also defined \textsc{Forbidden}($i$) to capture the notion of \textit{forbidden} in \cite{Garg2020} (discussed in \Cref{section:preliminaries}).

\begin{center}
    \begin{tabular}{|l|l|}
        \hline
        Macro & What it stands for\\
        \hline
        $\mathcal{D}$ & $\{i\in V(G): st.i =$ $IN\}$.\\
        \textsc{Satisfied}($i$) & $st.i=IN\lor(\forall d\in D_i, \exists j\in Adj_i: d\in S_j\land st.j=IN)$.\\
        \textsc{Unsatisfied-DS}$(i)$ & $\lnot$\textsc{Satisfied}($i$).\\
        \textsc{Removable-DS}$ (i$) & $(\forall d \in D_i : (\exists j \in Adj_i: d \in S_j \land st.j=$ $IN))\land$\\
         & \quad $(\forall j \in Adj_i,\forall~d \in D_j:d\in S_i\implies$\\
         & \quad \quad $(\exists k \in Adj_j, k\neq i:(d \in S_k \land st.k =$ $IN)))$.\\
        \textsc{Dominators-Of}($i$) & $\{j\in Adj_i, st.j=IN:\exists d\in D_i:d\in S_j\}\cup\{i\}$\quad if $st.i=IN$\\
         & $\{j\in Adj_i, st.j=IN:\exists d\in D_i:d\in S_j\}$\quad \quad \quad \quad otherwise.\\
        \textsc{Forbidden-DS}$(i)$ & $st.i=IN\land$ \textsc{Removable-DS}$(i)\land$\\
         & \quad $(\forall j \in Adj_i,\forall~d \in D_j:d\in S_i\implies$\\
         & \quad \quad $((\forall k \in$ \textsc{Dominators-Of}$(j)$, $k\neq i:(d \in S_k \land st.k =$ $IN))\implies$\\
         & \quad \quad \quad $(id.k<id.i\lor\lnot$\textsc{Removable-DS}$(k))))$.\\
        \hline
    \end{tabular}
\end{center}

The general idea our algorithm is as follows. 
\begin{enumerate}
    \item We allow a node to enter the dominating set unconditionally if it is unsatisfied, i.e., \textsc{Satisfied}($i$) is false. This ensures that $G$ enters a feasible state (where $\mathcal{D}$ is a dominating set) as quickly as possible.
    \item While entering the dominating set is not coordinated with others, leaving the dominating set is coordinated with neighboring nodes. 
    Node $i$ can leave the dominating set only if it is removable. But before it does that, it needs to coordinate with others so that too many nodes do not leave, creating a race condition. 
    Specifically, if $i$ serves for a demand $d$ in $D_j$ where $j \in Adj_i$ and the same demand is also served by another node $k$ ($k\in Adj_j$) then $i$ leaves only if (1) $id.k < id.i$ or (2) $k$ is not removable. 
    This ensures that if some demand $d$ of $D_j$ is satisfied by both $i$ and $k$ both of them cannot leave the dominating set simultaneously. 
    This ensures that $j$ will remain dominated. 
\end{enumerate} 
Thus, the rules for \Cref{algorithm:rules-ds} are as follows: 

\begin{algorithm}\label{algorithm:rules-ds}Rules for node $i$.
    \begin{center}
        \begin{tabular}{|l|}
            \hline
            \textsc{Forbidden-DS}$(i)\longrightarrow st.i=OUT$.\\
            \textsc{Unsatisfied-DS}$(i)\longrightarrow st.i=IN$.\\
            \hline
        \end{tabular}
    \end{center}
\end{algorithm}

We decompose  \Cref{algorithm:rules-ds} into two parts: (1) \Cref{algorithm:rules-ds}.1, that only consists of first guard and action of \Cref{algorithm:rules-ds} and (2) \Cref{algorithm:rules-ds}.2, that only consists of the second guard and action of \Cref{algorithm:rules-ds}. We use this decomposition in some of the following parts of this article section to relate the algorithm to eventual lattice linearity.

\section{Lattice Linear Characteristics of \Cref{algorithm:rules-ds}}\label{section:sdds}

In this section, we analyze the characteristics of \Cref{algorithm:rules-ds} to demonstrate that it is eventually lattice linear. We proceed as follows.
In \Cref{subsection:ds-propositions}, we state the propositions which define the feasible and optimal states of the SDDS problem, along with some other definitions. In \Cref{subsection:guarantee-feasible}, we show that $G$ reaches a state where it manifests a (possibly non-minimal) dominating set.
In \Cref{subsection:ds-action2}, we show that after when $G$ reaches a feasible state, \Cref{algorithm:rules-ds} behaves like a lattice linear algorithm. 
In \Cref{subsection:termination}, we show that when $\mathcal{D}$ is a minimal dominating set, no nodes are enabled. 
In \Cref{subsection:ds-eventual}, we argue that because there is a bound on interference between Algorithm 1.1 and 1.2 even when the nodes read old values, \Cref{algorithm:rules-ds} is an Eventually Lattice Linear Self-Stabilizing (ELLSS) algorithm.
In \Cref{subsection:time-space-complexity-analysis}, we study the time and space complexity attributes of \Cref{algorithm:rules-ds}.

\subsection{Propositions stipulated by the SDDS problem}\label{subsection:ds-propositions}

Notice that the SDDS problem stipulates that the nodes whose state is $IN$ must collectively form a dominating set. Formally, we represent this proposition as $\mathcal{P}_d^\prime$ which is defined as follows.
\begin{center}
    $\mathcal{P}_d^\prime(\mathcal{D}) \equiv \forall i\in V(G):(i\in \mathcal{D}\lor (\forall d\in D_i,\exists j\in Adj_i: (d\in S_j\land j\in \mathcal{D})))$.
\end{center}
The SDDS problem stipulates an additional condition that $\mathcal{D}$ should be a minimal dominating set. We formally describe this proposition $\mathcal{P}_d$ as follows.

\begin{center}
    $\mathcal{P}_d(\mathcal{D})\equiv \mathcal{P}^\prime_d(\mathcal{D})\land(\forall i\in \mathcal{D}, \lnot\mathcal{P}_d^\prime(\mathcal{D}\setminus\{i\}))$.
\end{center}

If $\mathcal{P}_d^\prime(\mathcal{D})$ is true, then $G$ is in a feasible state. And, if $\mathcal{P}_d(\mathcal{D})$ is true, then $G$ is in an optimal state. 

Based on the above definitions , we define two scores with respect to the global state, $RANK$ and $BADNESS$.
$RANK$ determines the number of nodes needed to be added to $\mathcal{D}$ to change $\mathcal{D}$ to a dominating set. 
$BADNESS$ determines the number of nodes that are needed to be removed from $\mathcal{D}$ to make it a minimal dominating set, given that $\mathcal{D}$ is a (possibly non-minimal) dominating set. Formally, we define $RANK$ and $BADNESS$ as follows.

\begin{definition}
    $RANK(\mathcal{D})\equiv\min\{|\mathcal{D}^\prime|-|\mathcal{D}|:\mathcal{P}^\prime_d(\mathcal{D}^\prime)\land \mathcal{D}\subseteq \mathcal{D}^\prime\}$.
\end{definition}

\begin{definition}
    $BADNESS(\mathcal{D})\equiv\max\{|\mathcal{D}|-|\mathcal{D}^\prime|:\mathcal{P}^\prime_d(\mathcal{D}^\prime)\land \mathcal{D}^\prime\subseteq \mathcal{D}\}$.
\end{definition}

\subsection{Guarantee to Reach a Feasible State by \Cref{algorithm:rules-ds}.2}\label{subsection:guarantee-feasible}

In this subsection, we show that if the nodes execute \Cref{algorithm:rules-ds}.2 only, then $G$ is guaranteed to reach a \textit{feasible} state where $\mathcal{D}$ is a (possibly non-minimal) dominating set. 

\iffalse
As per our current problem, we define \textsc{Unsatisfied} and \textsc{Forbidden} according to $\mathcal{P}_d$ as follows.

\begin{definition}
    \textsc{Unsatisfied}$(G,i,\mathcal{P}_d)\equiv (st.i=OUT)\land (\exists d\in D_i: (\forall j\in Adj_i, d\not\in S_j))$.
\end{definition}

\begin{definition}
    \textsc{Forbidden}$(G,i,\mathcal{P}_d)\equiv$ $st.i=IN\land$ \textsc{Removable}$(i)\land$\\
        \quad $(\forall j \in Adj_i,\forall~d \in D_j:d\in S_i\implies$\\
        \quad \quad $((\forall k \in$ \textsc{Dominators-Of}$(j)$, $k\neq i:(d \in S_k \land st.k =$ $IN))\implies$\\
        \quad \quad \quad $((id.k<id.j)\lor(id.k>id.j\land\lnot$\textsc{Removable}$(k)))))$.
    %\textsc{Forbidden}$(G,i,\mathcal{P}_d)\equiv$ $(st.i=IN)$ $\land$ $($\textsc{Removable}$(i))\land (\forall j\in I_i\lnot$\textsc{Removable}$(i))$.
\end{definition}

\begin{lemma}\label{lemma:d-not-ds}
    If $t.\mathcal{D}$ is not a dominating set at the beginning of some round $t$, then $(t+1).\mathcal{D}$ is a dominating set at the beginning of round $t+1$.
\end{lemma}
\fi

\begin{lemma}\label{lemma:d-not-ds}
    Let $t.\mathcal{D}$ be the value of $\mathcal{D}$ at the beginning of round $t$. 
If $t.\mathcal{D}$ is not a dominating set then $(t+1).\mathcal{D}$ is a dominating set.
\end{lemma}

\begin{proof}
    Let $i$ be a node such that $i\in t.\mathcal{D}$ and $i\not\in(t+1).\mathcal{D}$, {i.e., $i$ leaves the dominating set in round $t$}. This means that $i$ will remain satisfied and each node in $Adj_i$ is satisfied, even when $i$ is removed. This implies that $i$ will not reduce the feasibility of $t.D$; it will not increase the value of $RANK$.
    
    Now let $\ell$ be a node such that $\ell\not \in t.\mathcal{D}$ which is not satisfied when it evaluates its guards in round $t$. This implies that $\exists~d\in D_\ell$ such that $d$ is not present in $S_j$ for any $j\in Adj_\ell$. According to the algorithm, the guard of the second action is true for $\ell$. This implies that $st.\ell$ will be set to $IN$. 
    
    It can also be possible for the node $\ell$ that it is satisfied when it evaluates its guards in round $t$. This may happen if some other nodes around $\ell$ already decided to move to $\mathcal{D}$, and as a result $\ell$ is now satisfied. Hence $\ell\not\in(t+1).\mathcal{D}$ and we have that $\ell$ is dominated at round $t+1$.
    
    Therefore, we have that $(t+1).\mathcal{D}$ is a dominating set, which may or may not be minimal.
    \qed
\end{proof}

By \Cref{lemma:d-not-ds}, we have that if the $G$ is in a state where $RANK >0$ then by the next round, $RANK$ will be $0$.

\subsection{Lattice Linearity of \Cref{algorithm:rules-ds}.1}\label{subsection:ds-action2}

In the following lemma, we show that \Cref{algorithm:rules-ds}.1 is lattice linear.

\begin{lemma}\label{lemma:ds}
    If $t.\mathcal{D}$ is a non-minimal dominating set then according to \Cref{algorithm:rules-ds} (more specifically, \Cref{algorithm:rules-ds}.1),
    there exists at least one node such that $G$ cannot reach a minimal dominating set until that node is removed from the dominating set.
\end{lemma}

\begin{proof}
Since $\mathcal{D}$ is a dominating set, we have that the second guard is not true for any node in $G$.

Since $\mathcal{D}$ is not minimal, there exists at least one node that must be removed in order to make $\mathcal{D}$ minimal. Let $S^\prime$ be the set of nodes which are removable. 
Let $M$ be some node in $S^\prime$. If $M$ is not serving any node, then \textsc{Forbidden}($M$) is trivially true. Otherwise there exists at least one node $j$ which is served by $M$, that is, $\exists d\in D_j:d\in S_M$. We study two cases which are as follows:
(1) for some
node $j$ served by $M$, there does not exist a node $b \in S^\prime$ which serves $j$, and (2) for any node $b \in S^\prime$ such that $M$ and $b$ serve some common node $j$, $id.b<id.M$. 

In the first case, $M$ cannot be removed because \textsc{Removable}($M$) is false and, hence, $M$ cannot be in $S^\prime$, thereby leading to a contradiction.
In the second case, \textsc{Forbidden}($M$) is true and \textsc{Forbidden}($b$) is false since $id.b < id.M$. Thus, node $b$ cannot leave the dominating set in that $b$ cannot leave until $M$ leaves. In both the cases, we have that $j$ stays dominated.

Since ID of every node is distinct, we have that there exists at least one node $M$ for which \textsc{Forbidden}($M$) is true. For example, \textsc{Forbidden}($M$) is true for the node with the highest ID in $S^\prime$; $G$ cannot reach a minimal dominating set until $M$ is removed from the dominating set.
\qed
\end{proof}

From \Cref{lemma:ds}, it follows that \Cref{algorithm:rules-ds}.1 satisfies the condition of lattice linearity defined in \Cref{section:preliminaries}. It follows that if we start from a state where $\mathcal{D}$ is a (possibly non-minimal) dominating set and execute \Cref{algorithm:rules-ds}.1 then it will reach  a state where $\mathcal{D}$ is a minimal dominating set even if nodes are executing with old information about others. Next, we have the following result which follows from \Cref{lemma:ds}. 

\begin{lemma}
    Let $t.\mathcal{D}$ be the value of $\mathcal{D}$ at the beginning of round $t$. 
    If $t.\mathcal{D}$ is a non-minimal dominating set then $|(t+1).\mathcal{D}| \leq |t.\mathcal{D}|-1$, and $(t+1).\mathcal{D}$ is a dominating set.
\end{lemma}

\begin{proof}
From \Cref{lemma:ds}, {at least one node $M$ (including the maximum ID node in $S^\prime$} from the proof of \Cref{lemma:ds}) would be removed  in round $t$. Furthermore, since $\mathcal{D}$ is a dominating set, \textsc{Unsatisfied}($i$) is false at every node $i$. Thus, no node is added to $\mathcal{D}$ in round $t$. Thus, the $|(t+1).\mathcal{D}|\leq |t.\mathcal{D}|-1$. 
    
    For any node $M$ that is removable, \textsc{Forbidden}($i$) is true only if any node $j$ which is (possibly) served by $M$ has other neighbors (of a lower ID) which serve the demands which $M$ is serving to it. This guarantees that $j$ stays dominated and hence $(t+1).\mathcal{D}$ is a dominating set.
    \qed
\end{proof}

\subsection{Termination of \Cref{algorithm:rules-ds}}\label{subsection:termination}

The following lemma studies the action of \Cref{algorithm:rules-ds} when $\mathcal{D}$ is a minimal dominating set.

\begin{lemma}\label{lemma:ds-minimal}
Let $t.\mathcal{D}$ be the value of $\mathcal{D}$ at the beginning of round $t$. If $\mathcal{D}$ is a minimal dominating set, then $(t+1).\mathcal{D}=t.\mathcal{D}$.

\end{lemma}

\begin{proof}
    Since $\mathcal{D}$ is a dominating set, \textsc{Satisfied}($i$) is true for every node in $V(G)$. Hence, the second action is disabled for every node in $V(G)$.
Since $\mathcal{D}$ is minimal, \textsc{Removable}($i$) is false for every node in $\mathcal{D}$. Hence, the first action is disabled at every node $i$ in $\mathcal{D}$. 
Thus, $\mathcal{D}$ remains unchanged.
    \qed
\end{proof}

\subsection{Eventual Lattice Linearity of \Cref{algorithm:rules-ds}}\label{subsection:ds-eventual}

\Cref{lemma:ds} showed that \Cref{algorithm:rules-ds}.1 is lattice linear. 
In this subsection, we make additional observations about \Cref{algorithm:rules-ds} to generalize the notion of lattice linearity to eventually lattice linear algorithms. 
We have the following observations.
\begin{enumerate}
    \item From \Cref{lemma:d-not-ds}, starting from any state, \Cref{algorithm:rules-ds}.2 will reach a feasible state even if a node reads old information about the neighbors. This is due to the fact that \Cref{algorithm:rules-ds}.2 only adds nodes to $\mathcal{D}$.
    \item From \Cref{lemma:ds}, if we start $G$ in a feasible state where no node has incorrect information about the neighbors in the initial state then \Cref{algorithm:rules-ds}.1 reaches a minimal dominating set. Note that this claim remains valid even if the nodes execute actions of \Cref{algorithm:rules-ds}.1 with old information about the neighbors as long as the initial information they use is correct. 
    \item Now, we observe that \Cref{algorithm:rules-ds}.1 and \Cref{algorithm:rules-ds}.2 have very limited interference with each other, and so an arbitrary graph $G$ will reach an optimal state even if nodes are using old information. 
    In this case, observe that any node $i$ can execute the action of the first guard incorrectly at most once. After going $OUT$ incorrectly, when it reads the correct information about other nodes, then it will execute the guard of the second action and change its state to $IN$, after which, it can go out only if it evaluates that \textsc{Forbidden}($i$) is true.
\end{enumerate}

From the above observations, if we allow the nodes to read old values, then the nodes can violate the feasibility of $G$ finitely many times and so $G$ will eventually reach a feasible state and stay there forever. 
We introduce the term Eventually Lattice Linear Self-stabilizing algorithms (ELLSS).
Before defining ELLSS algorithms, we define the class of Lattice Linear Algorithms (LL) as follows.

\begin{definition}\textbf{Lattice Linear Algorithms}.
    LL algorithms are the algorithms under which a system $G$ is forced to traverse a lattice of states and proceed to reach an optimal solution.
\end{definition}

Note that LL is a generalization of the notion of lattice linearity introduced in \cite{Garg2020}. In \cite{Garg2020}, the existence of the lattice arises from the problem at hand. In LL, it arises by constraints imposed by the algorithm.  

The class of ELLSS algorithms can be defined as follows.

\begin{definition}\label{definition:ellss} \textbf{Eventually Lattice Linear Self-Stabilizing Algorithms.}
An algorithm is ELLSS if its rules can be split into $F_1$ and $F_2$ and there exists a predicate $\mathcal{R}$ such that 
    \begin{enumerate}[label=(\alph*)]
        \item Any computation of $F_1 [] F_2$, that is, the union of the actions in $F_1$ and $F_2$,
        eventually reaches a state
        where $\mathcal{R}$ is stable in $F_1 [] F_2$,
        i.e., $\mathcal{R}$ is true and remains true subsequently (even if the nodes read old values of other nodes),
    \item $F_2$ is an LL algorithm, given that it starts in a state in $\mathcal{R}$,
        \item Actions in $F_1$ are disabled once the program reaches $\mathcal{R}$.
    \end{enumerate}
\end{definition}

In \Cref{algorithm:rules-ds}, $F_1$ corresponds to \Cref{algorithm:rules-ds}.2 and $F_2$ corresponds to \Cref{algorithm:rules-ds}.1.
And, the above discussion shows that this algorithm satisfies the properties of \Cref{definition:ellss} and $\mathcal{R}$ corresponds to the predicate that $\mathcal{D}$ is a dominating set, i.e., $\mathcal{R}\equiv\mathcal{P}_d^\prime(\mathcal{D})$. 

Finally, we note that in \Cref{algorithm:rules-ds}, we chose to be aggressive for a node to enter the dominating set but cautious to leave the dominating set.
A similar ELLSS SDDS algorithm is feasible where a node is cautious to enter the dominating set but aggressive to leave it.

\subsection{Analysis  of \Cref{algorithm:rules-ds}: Time and Space complexity}\label{subsection:time-space-complexity-analysis}

\begin{theorem}\label{theorem:ds-convergence-time}
    \Cref{algorithm:rules-ds} converges in $2n$ moves.
\end{theorem}

\begin{proof}
    In the beginning of the algorithm, let $G$ be the input graph. Assume for contradiction that $\mathcal{D}$ is not a dominating set as per the input graph $G$. Let $V^\prime$ be the set of nodes, $V^\prime\subseteq V(G)\setminus \mathcal{D}$ such that each node $i$ in $V^\prime$, $i$ is not satisfied. It means that the guard of the second action holds true for $i$. Therefore, we have that $i$ will execute the second action during the first round (or within first $n$ moves) and $G$ will enter a feasible state.
    
    The above observation implies that by the end of the first round, $\mathcal{D}$ becomes a dominating set. $\mathcal{D}$ may or may not be a minimal dominating set. So in each subsequent time step, $\mathcal{D}$ will reduce in size by at least one node (by \Cref{lemma:ds}) but still be a valid dominating set. So we have that \Cref{algorithm:rules-ds} will converge within $n$ moves after when $G$ enters a feasible state.
    
\end{proof}

\begin{corollary}\label{corollary:converge-1-round-and-n-moves}
    $\mathcal{D}$ will be feasible within 1 round. After entering a feasible state, $\mathcal{D}$ will be optimal within $n$ moves.
\end{corollary}

\begin{corollary}\label{corollary:algo-stabilizing-silent}
    \Cref{algorithm:rules-ds} is self-stabilizing and silent.
\end{corollary}

\begin{lemma}\label{lemma:time-complexity-1-time-step}
    At any time-step, a node will take $O((\Delta)^4\times (max_d)^2)$ time, where
    \begin{enumerate}
        \item $\Delta$ is the maximum degree of any node in $V(G)$,
        \item $max_d$ is the total number of distinct demands made by all the nodes in $V(G)$.
    \end{enumerate}
\end{lemma}

\begin{proof}
    There are 3 expressions in the guard which are separated by an ``and'' ($\land$) operator.
    
    \textbf{First Expression}: ``$st.i=IN$'' takes constant amount of time.
    
    \textbf{Second expression}: Second expression is a \textsc{Removable} macro, which is a conjunction of two expressions. The first expression contains a universal quantifier ($\forall d\in D_i$) nested with an existential quantifier ($\exists j\in Adj_i$). The second expression contains two universal quantifiers ($\forall j\in Adj_i$ and $\forall d\in D_j$) and an existential quantifier ($\exists k\in Adj_j$) nested one after the other. Therefore, the time complexity of this expression can be evaluated to be of the time complexity of $\Delta^2\times max_d$.
    
    \textbf{Third expression}: The second expression, there are 4 objects as follows.
    \begin{enumerate}
        \item A universal quantifier - $\forall j\in Adj_i$.
        \item A universal quantifier - $\forall d\in D_j$.
        \item A universal quantifier - $\forall k\in$ \textsc{Domintors-Of}($j$).
        \item The \textsc{Removable} macro.
    \end{enumerate}
    Therefore, the time complexity to compute the expression evaluates to be in the order of the product of $(|D_i|)^2$, and $(|Adj_i|)^4$, which is upper bounded by $O(\Delta^4\times (max_d)^2)$.
    
    The total time complexity to evaluate the first guard computes to be $O((\Delta)^4\times (max_d)^2)$.
    
    Assuming that the services (demands) are stored in any node in an array of length equal to the total number of services (demands), and the assuming that the index of the service is equal to the encoding of the service, containing 1 at the corresponding index iff the corresponding service (demand) is being provided (demanded) by that node, we can compute the presence of a service (demand) in that node in $O(1)$ time. Therefore, we can compute $d\in S_k$ or $d\in S_j$ in a constant amount of time.
    
    The second guard has the time complexity of $O(\Delta\times max_d)$. Therefore, it does not affect the order of the resultant time-complexity.
\end{proof}

\begin{corollary}\label{corollary:storage-each-node}
    The space required to store the services and demands in each node is $O(max_d)$, where $max_d$ is the total number of distinct demands made by all the nodes in $V(G)$.
\end{corollary}

\section{Other Examples}\label{section:other-examples}

The sequence of states of $G$ under \Cref{algorithm:rules-ds} is essentially divided into two phases: (1) the system entering a feasible state (reduction of $RANK$ to zero), and then (2) the system entering an optimal state (reduction of $BADNESS$ to zero). \Cref{algorithm:rules-ds} first takes the system to a feasible state where $RANK=0$ and then it takes the system to an optimal state where $RANK=0\land BADNESS=0$. 

This notion was used to define the concept of ELLSS algorithms. The notion of ELLSS algorithms can be extended to numerous other problems 
where the optimal global state can be defined in terms of a minimal (or maximal) set $\mathcal{S}$ of nodes.
This includes the vertex cover problem, independent set problem and their variants.
Once the propositions regarding $\mathcal{S}$ is defined where its structure depends on some relation of nodes with their neighbours, the designed algorithm can decide which node to put $IN$ the set and which nodes to take $OUT$.

In this section, we describe algorithms for vertex cover and independent set, along with graph coloring, which follow from the structure that we have laid for the SDDS problem. Thus, the algorithms  we describe for these problems are also ELLSS algorithms. The proofs of correctness follow from the proofs of correctness described above for the SDDS problem.

\subsection{Vertex cover}

In the \textit{vertex cover} (VC) problem, the input is an arbitrary graph $G$, and the task is to compute a minimal set $\mathcal{V}$ such that for any edge $\{i,j\}\in E(G)$, $(i\in \mathcal{V})\lor (j\in \mathcal{V})$. If a node $i$ is in $\mathcal{V}$, then $st.i=IN$, otherwise $st.i=OUT$. To develop an algorithm for VC, we utilize the macros in the following table.

\begin{center}
    \begin{tabular}{|l|l|}
        \hline
        \textsc{Removable-VC}$(i$) & $(\forall j \in Adj_i, st.j=IN)$.\\
        \textsc{Unsatisfied-VC}$(i)$ & $(st.i=OUT)\land(\exists j\in Adj_i:st.j=OUT)$.\\
        \textsc{Forbidden-VC}$(i)$ & $(st.i=IN)$ $\land$ $($\textsc{Removable-VC}$(i))\land$\\
         & \quad $(\forall j\in Adj_i: (id.j<id.i)\lor \lnot$\textsc{Removable-VC}$(j))$. \\
        \hline
    \end{tabular}
\end{center}

The proposition $\mathcal{P}^\prime_v$ defining a feasible state and the proposition $\mathcal{P}_v$ defining the optimal state can be defined as follows.
\begin{center}
    $\mathcal{P}_v^\prime(\mathcal{V})\equiv \forall i\in V(G):((i\in \mathcal{V})\lor (\forall j\in Adj_i, j\in \mathcal{V}))$.\\
    $\mathcal{P}_v(\mathcal{V})\equiv \mathcal{P}_v^\prime(\mathcal{V}) \land (\forall i\in \mathcal{V}, \lnot\mathcal{P}_v^\prime(\mathcal{V}- \{i\})).$
\end{center}
Based on the definitions above, the algorithm for VC is described as follows.
\begin{algorithm}\label{algorithm:rules-vc}Rules for node $i$.
    \begin{center}
        \begin{tabular}{|l|}
            \hline
            \textsc{Forbidden-VC}$(i)\longrightarrow st.i=OUT$.\\
            \textsc{Unsatisfied-VC}$(i)\longrightarrow st.i=IN$.\\
            \hline
        \end{tabular}
    \end{center}
\end{algorithm}

Once again, this is an ELLSS algorithm in that it satisfies the conditions in \Cref{definition:ellss}, where $F_1$ corresponds to the second action of \Cref{algorithm:rules-vc}, $F_2$ corresponds to its first action, and $\mathcal{R}\equiv \mathcal{P}_v^\prime(\mathcal{V})$.
Thus, starting from any arbitrary state, the algorithm eventually reaches a state where $\mathcal{V}$ is a minimal vertex cover.

Note that in \Cref{algorithm:rules-vc}, the definition of \textsc{Removable} relies only on the information about distance-1 neighbors. Hence, the evaluation of guards take $O(\Delta^3)$ time. In contrast, (the standard) dominating set problem requires information of distance-2 neighbors to evaluate \textsc{Removable}. Hence, the evaluation of guards in that would take $O(\Delta^4)$ time. 

\subsection{Independent set}

In VC and SDDS problems, we tried to reach a minimal set. Here on the other hand, we have to obtain a maximal set. In the \textit{independent set} (IS) problem, the input is an arbitrary graph $G$, and the task is to compute a maximal set $\mathcal{I}$ such that for any two nodes $i\in\mathcal{I}$ and $j\in\mathcal{I}$, if $i\neq j$, then $\{i,j\}\neq E(G)$.

The proposition $\mathcal{P}^\prime_i$ defining a feasible state and the proposition $\mathcal{P}_i$ defining the optimal state can be defined as follows.
\begin{center}
    $\mathcal{P}_i^\prime(\mathcal{I})\equiv \forall i\in V(G):((i\not\in\mathcal{I})\lor (\forall j\in Adj_i: j\not\in \mathcal{I}))$.\\
    $\mathcal{P}_i(\mathcal{I})\equiv \mathcal{P}_i^\prime(\mathcal{I})\land(\forall i\in V(G)\setminus\mathcal{I}, \lnot\mathcal{P}_i^\prime(\mathcal{I}\cup\{i\}))$.
\end{center}

If a node $i$ is in $\mathcal{I}$, then $st.i=IN$, otherwise $st.i=OUT$. To develop the algorithm for independent set, we define the macros in the following table.

\begin{center}
    \begin{tabular}{|l|l|}
        \hline
        %$O_i$ & the nodes $j$ in $V(G)$ such that $id.j > id.i \land st.j =$ $OUT$.\\
        \textsc{Addable}($i$) & %\quad \quad $\land$
         $(\forall j \in Adj_i, st.j=OUT)$.\\
        \textsc{Unsatisfied-IS}$(i)$ & $(st.i=IN)\land(\exists j\in Adj_i:st.j=IN$).\\
        \textsc{Forbidden-IS}$(i)$ & $st.i=OUT\land$ \textsc{Addable}$(i)\land$\\
         & \quad $(\forall j\in Adj_i:((id.j<id.i)\lor(\lnot$\textsc{Addable}$(j)))$.\\
        %& all the edges incident to i are being covered\\
        \hline
    \end{tabular}
\end{center}

Based on the definitions above, the algorithm for IS is described as follows.
\begin{algorithm}\label{algorithm:rules-is}Rules for node $i$.
    \begin{center}
        \begin{tabular}{|l|}
            \hline
            \textsc{Forbidden-IS}$(i)\longrightarrow st.i=IN$.\\
            \textsc{Unsatisfied-IS}$(i)\longrightarrow st.i=OUT$.\\
            \hline
        \end{tabular}
    \end{center}
\end{algorithm}

This algorithm is an ELLSS algorithm as well: as per \Cref{definition:ellss}, $F_1$ corresponds to the second action of \Cref{algorithm:rules-vc}, $F_2$ corresponds to its first action, and $\mathcal{R}\equiv \mathcal{P}_i^\prime(\mathcal{I})$.
Thus, starting from any arbitrary state, the algorithm eventually reaches a state where $\mathcal{I}$ is a maximal independent set.

In \Cref{algorithm:rules-is}, the definition of \textsc{Addable} relies only on the information about distance-1 neighbors. Hence, the evaluation of guards take $O(\Delta^3)$ time.

\subsection{Coloring}

In this section, we extend ELLSS algorithms to graph coloring. 
In the \textit{graph coloring} (GC) problem, the input is a graph $G$ and the task is to assign colors to all the nodes of $G$ such that no two adjacent nodes have the same color.

Unlike vertex cover, dominating set or independent set, coloring does not have a binary domain. Instead, we correspond the equivalence of changing the state to $IN$ to the case where a node increases its color. And, the equivalence of changing the state to $OUT$ corresponds to the case where a node decreases its color. With this intuition, we define the macros as shown in the following table. 

\begin{center}
    \begin{tabular}{|l|l|}
        \hline
        \textsc{Conflicted}($i$) & $(\exists j \in Adj_i:(color.j=color.i))$.\\
        \textsc{Subtractable}($i$) & $\exists c\in [1:color.i-1]: ((\forall j\in Adj_i: color.j\neq c))$.\\
        \textsc{Unsatisfied-GC}$(i)$ & \textsc{Conflicted}$(i)$.\\
        \textsc{Forbidden-GC}$(i)$ & $\lnot$\textsc{Conflicted}($i$) $\land$ \textsc{Subtractable}$(i)\land$\\
         & $(\forall j\in Adj_i:(id.j<id.i\lor \lnot$\textsc{Subtractable}$(j)))$.\\
        \hline
    \end{tabular} 
\end{center}

The proposition $\mathcal{P}^\prime_c$ defining a feasible state and the proposition $\mathcal{P}_c$ defining an optimal state is defined below. $\mathcal{P}_c$ is true when all the nodes have lowest available color, that is, for any node $i$ and for all colors $c$ in $[1:deg(i)+1]$, either $c$ should be greater than $color.i$ or $c$ should be equal to the color of one of the neighbors $j$ of $i$.

\begin{center}
    $\mathcal{P}_c^\prime\equiv \forall i\in V(G),\forall j\in Adj_i:color.i\neq color.j$.\\
    $\mathcal{P}_c\equiv \mathcal{P}_c^\prime\land (\forall i\in V(G):(\forall c\in [1:color.i-1]:$\\
    $(c<color.i\implies (\exists j\in Adj_i: color.j=c))))$.
\end{center}

Unlike SDDS, VC and IS, in graph coloring (GC), each node is associated with a variable $color$ that can take several possible values (the domain can be as large as the set of natural numbers). 
As mentioned above, the action of increasing the color is done whenever a conflict is detected. However, decreasing the color is achieved only with coordination with others. Thus, the actions of the algorithm are shown in \Cref{algorithm:rules-c}. 

\iffalse
The notion of $RANK$ and badness for coloring is different than than that for SDDS, VC or IS. So we defined those terms with respect to GC as follows.
\begin{definition}
    $RANK\equiv|\{i\in V(G):\exists j\in Adj_i, color.i=color.j\}|$.
\end{definition}

\begin{definition}
    $BADNESS\equiv|\{i\in V(G):\exists c\in [1:deg(i)+1]:(c<color.i\implies(\forall j\in Adj_i:color.j\neq c)\}|$.
\end{definition}
\fi
\begin{algorithm}\label{algorithm:rules-c}Rules for node $i$.
    \begin{center}
        \begin{tabular}{|l|}
            \hline
            \textsc{Forbidden-GC}$(i)\longrightarrow color.i=\min\limits_{c}\{c\in [1:color.i-1]:(\forall j\in Adj_i: color.j\neq c)\}$.\\
            \textsc{Unsatisfied-GC}$(i)$ $\longrightarrow color.i=color.i+id.i$.\\
            \hline
        \end{tabular}
    \end{center}
\end{algorithm}

\Cref{algorithm:rules-c} is an ELLSS algorithm: according to \Cref{definition:ellss}, $F_1$ corresponds to the second action of \Cref{algorithm:rules-vc}, $F_2$ corresponds to its first action, and $\mathcal{R}\equiv \mathcal{P}_c^\prime$.
Thus, starting from any arbitrary state, the algorithm eventually reaches a state where no two adjacent nodes have the same color and no node can reduce its color.

\section{Conclusion}\label{section:conclusion}

We extended lattice linear algorithms from \cite{Garg2020} to the context of self-stabilizing algorithms. The approach in \cite{Garg2020} relies on the assumption that the algorithm starts in  specific initial states and, hence, it is not directly applicable in self-stabilizing algorithms. A key benefit of lattice linear algorithms is that correctness is preserved even if nodes are reading old information about other nodes. Hence, they allow a higher level of concurrency. 

We began with the service-demand based dominating set (SDDS) problem and designed a self-stabilizing algorithm for the same. Subsequently, we observed that it consists of two parts: One part is a lattice linear algorithm that constructs a minimal dominating set if it starts in some valid initial states, say $\mathcal{R}$. The second part makes sure that it gets the program in a state where $\mathcal{R}$ becomes true and stays true forever. We also showed that these parts can only have bounded interference thereby guaranteeing that the overall program is self-stabilizing even if the nodes read old values of other nodes. 

We introduced the notion of eventually lattice-linear self-stabilization to capture such algorithms. We demonstrated that it is possible to develop eventually lattice linear self-stabilizing (ELLSS) algorithms for vertex cover, independent set and graph coloring. 

We note that Algorithms \ref{algorithm:rules-ds}-\ref{algorithm:rules-c} could also be designed to be lattice linear self-stabilizing algorithms (LLSS) if we change the second action of these algorithms to account for the neighbors in the same fashion as done for the second action. 

The algorithms \ref{algorithm:rules-ds}-\ref{algorithm:rules-c} converge under central, distributed, or synchronous daemon. Due to the property of ELLSS, its straightforward implementation in read/write model is also self-stabilizing.
Intuitively, in a straightforward translation, each remote variable is replaced by a local copy of that variable and this copy is updated asynchronously. Normally, such straightforward translation into read/write atomicity does not preserve self-stabilization. However, the ELLSS property of the self-stabilization ensures correctness of the straightforward translation.

As future work, an interesting direction can be to study which class of problems can the paradigm of LL and ELLSS algorithms be extended to. Also, it is interesting to study if we can implement approximation algorithms under these paradigms.

\bibliography{dk2d1.bib}
\bibliographystyle{acm}

\end{document}